\documentclass{llncs}

\usepackage{nicefrac}
\usepackage{amssymb}
\usepackage{amsmath}

\title{A Simplest Undecidable Modal Logic
\thanks{Supported in part by  NSF grant IIS-0713061, the DAAD postdoc program, and by a Friedrich Wilhelm Bessel Research Award. Work done in part while
the second author was at the Rochester Institute of Technology.}}

\author{Edith Hemaspaandra \and Henning Schnoor}

\institute{
  Department of Computer Science,
  Rochester Institute of Technology,
  Rochester, NY 14623, U.S.A.
  \email{eh@cs.rit.edu, hs@cs.rit.edu}
  \and
  Institute for Computer Science,
  Christian-Albrechts-Universit\"at zu Kiel, Kiel, Germany
}

\newcommand{\set}[1]{\ensuremath\left\{#1\right\}}
\newcommand{\card}[1]{\left| #1 \right|}
\newcommand{\var}[1]{\ensuremath{\mathrm{VAR}\!\left(#1\right)}}
\newcommand{\nmodels}{\ensuremath{\not\models}}
\newcommand{\sat}[1]{#1$-$\textrm{\sf{SAT}}}
\newcommand{\md}[1]{\mathsf{md}\left(#1\right)}
\newcommand{\mathtext}[1]{\ensuremath{\mathrm{\text{#1}}}}

\newcommand{\decisionproblem}[3]{
\medskip
\vspace*{1mm}
\begin{tabular}{ll}
\textit{Problem:} & #1 \\
\textit{Input:} & #2 \\
\textit{Question:} & #3
\end{tabular}
\smallskip
\vspace*{1mm}
}

\newcommand{\complexityclassname}[1]{\ensuremath{\mathrm{#1}}}
\newcommand{\complementclass}[1]{\complexityclassname{co}#1}
\newcommand{\RE}{\complexityclassname{RE}}
\newcommand{\CORE}{\complementclass{\RE}}
\newcommand{\NP}{\complexityclassname{NP}}

\newcommand{\logicname}[1]{\ensuremath{\mathsf{#1}}}

\newcommand{\K}[1]{\ensuremath{\logicname{K}(#1)}}

\newcommand{\F}{\mathcal F}
\newcommand{\ggrid}{\ensuremath{\mathrm{Global\text-Grid\text-Sat}}}

\newcommand{\formulaEQ}{\hat\varphi^\sim_{\mathtext{eq}}}
\newcommand{\formulaOneStep}{\hat\varphi_{1\mbox{-}\mathtext{step}}}
\newcommand{\formulaTwoStep}{\hat\varphi_{2\mbox{-}\mathtext{step}}}
\newcommand{\formulaEQOneTwoStep}{\hat\varphi^{\mathtext{grid}}}

\newcommand{\formulaUniversalWorld}{\hat\varphi^{\mathtext{univ}}}
\newcommand{\formulaEQOneTwoStepKernel}{\formulaEQOneTwoStep_{\mathtext{kernel}}}
\newcommand{\formulaEQOneTwoStepRefl}{\formulaEQOneTwoStep_{\mathtext{refl}}}
\newcommand{\formulaFinal}{\hat\varphi^{\mathtext{final}}}

\newcommand{\formulaRespectP}{\psi^{P}_{\mathtext{resp}}}
\newcommand{\formulaExistsNontrivialSuccessor}{\psi^{\mathtext{succ}}}
\newcommand{\gd}{_{\mathtext{grid}}}

\newcommand{\hvarphi}{\hat\varphi}
\newcommand{\hpsi}{\hat\psi}

\begin{document}

\maketitle

\begin{abstract}
Modal logics are widely used in computer science. The complexity of their satisfiability problems has been an active field of research since the 1970s. We prove that even very ``simple'' modal logics can be undecidable: We show that there is an undecidable modal logic that can be obtained by restricting the allowed models with a first-order formula in which only universal quantifiers appear.
\end{abstract}

\section{Introduction}

Modal logics are widely used in many areas of computer science. See, for example, \linebreak \cite{bazh05,frhuje02,codofl03,hamotu88,lare86,bega04,Moody-MODAL-LOGIC-FOR-DISTRIBUTED-COMPUTATION-CMU-2003,AucherBoellavdTorre-PRIVACY-WITH-MOCAL-LOGIC-DEON-2010}. The complexity of modal satisfiability problems has been an active field of research since Ladner's work in the 1970s~\cite{lad77}.  Early work focused on the complexity of single logics, but more recent work has focused on handling the computability and complexity for an infinite number of logics, see,  e.g., \cite{HemaspaandraSchnoor-Elementary-STACS-2008,DemriKonikowska-RELATIVE-SIMILARITY-LOGICS-DECIDABLE-JELIA-1998,anbene98}.

Our ultimate goal is to classify the complexity of ``all'' modal logics. We are particularly interested in elementary modal logics, i.e., modal logics whose models are defined by a first-order formula, since modal logics used in practice are often defined in this way. In addition, complexity analysis is often easier when looking at models rather than at axioms.

An important first step towards classifying the complexity of all elementary modal logics is to determine which modal logics are decidable. Our concrete goal in this paper is to find a ``simplest'' undecidable modal logic. 

A particularly interesting simple class of elementary modal logics are the universal elementary modal logics, in which the class of models is defined by a universal first-order formula.  Not only do many common modal logics belong to this class, it is also a class that is on the borderline of being decidable. In particular, it is known that many universal Horn elementary uni-modal logics are in PSPACE, and it is conjectured that all these logics are decidable~\cite{HemaspaandraSchnoor-Elementary-STACS-2008}. On the other hand, it is known that {\em global} satisfiability, i.e., satisfiability in all worlds in a model, is undecidable for universal elementary uni-modal logic~\cite{Hemaspaandra-PRICE-OF-UNIVERSALITY-NDJFC-1996}.

To show undecidability for modal logics, we typically need to be able to make the models look like a grid, and we need a way to access all the worlds in the grid, as well as a way to access the two direct successors of a world. An early example of this is that 2-dimensional temporal logic on $\mathbb N\times\mathbb N$ with modal operators $\bigcirc_u$ and $\bigcirc_r$ that access the direct ``up'' and ``right'' successors of a world, and $\Box_u$ and $\Box_r$ that access all worlds that are above or to the right of the current world is undecidable~\cite{Harel-RECURRING-DOMINOES-FCT-1983}. The undecidability also holds if $\Box_u$ and $\Box_r$ are the only modal operators~\cite{Spaan-NEXTTIME-NOT-NECESSARY-TARK-1990}. Other examples of undecidable modal logics include various logics of knowledge and time~\cite{HalpernVardi-COMPLEXITY-REASONING-KNOWLEDGE-TIME-JCSS-1989,LadnerReif-LOGIC-DISTRIBUTED-PROTOCOLS-TARK-1986,Spaan-NEXTTIME-NOT-NECESSARY-TARK-1990}, various logics that allow identification of worlds~\cite{bla-spa:j:modal-perspective,GarvonGoranko-MODAL-LOGIC-WITH-NAMES-JPL-1993}, and most products of transitive modal logics~\cite{Gabelaia}.

All these logics are multi-modal and certainly not universal. As will be explained in Section~\ref{sect:grid}, universal first-order formulas are not enough to enforce a grid-like structure. Still, there exists a first-order universal formula such that {\em global} satisfiability for uni-modal logics is undecidable ~\cite{Hemaspaandra-PRICE-OF-UNIVERSALITY-NDJFC-1996}. Section~\ref{sect:previous} will explain the idea behind the construction from~\cite{Hemaspaandra-PRICE-OF-UNIVERSALITY-NDJFC-1996} and why this construction can not be used to show that satisfiability for a universal elementary modal logic is undecidable.

However, using a much more complicated construction (see Section~\ref{sect:grid} for a high-level overview and the rest of Section~\ref{sect:proof} and the appendix for details), we will show that there exists an undecidable universal elementary uni-modal logic.  And this result holds even if we are not allowed to use the equality predicate in the first-order universal formula describing the models.   And so we have indeed found a ``simplest'' undecidable modal logic. 

\section{Preliminaries}\label{sect:preliminaries}

Modal logic syntactically extends propositional logic with an additional unary operator $\Box$ (a dual operator $\Diamond$ abbreviates $\neg\Box\neg$). The \emph{modal depth} of a formula $\varphi$, denoted with $\md\varphi$, is the maximal nesting degree of the $\Box$-operator in $\varphi$. A model for a modal formula is a graph with individual propositional assignments. To be precise, a \emph{frame} $F$ is a directed graph $(W,R),$ where the vertices in $W$ are called ``worlds.'' A \emph{model} $M=(W,R,\pi)$ consists of a frame $(W,R)$ and a function $\pi$ assigning to each variable the set of worlds in which $x$ is true. We say the model $M$ is \emph{based} on the frame $(W,R).$ If $\F$ is a class of frames, then a model is an $\F$-model if it is based on a frame in $\F$. We often write $w\in M$ instead of $w\in W$. For a world $w\in M,$ we define when a modal formula $\phi$ is \emph{satisfied} at $w$ in $M$ (written $M,w\models\phi$). If $\phi$ is a variable $x,$ then $M,w\models\phi$ if and only if $w\in\pi(x).$ As usual, $M,w\models\phi_1\wedge\phi_2$ if and only if $M,w\models\phi_1$ and $M,w\models\phi_2$, and $M,w\models\neg\phi$ iff $M,w\nmodels\phi.$ For the modal operator, $M,w\models\Box\phi$ if and only if $M,w'\models\phi$ for all worlds $w'$ with $(w,w')\in R$.

A standard way to define classes of frames is to use first-order formulas. The \emph{frame language} is the first-order language containing (in addition to Boolean operators) the binary relation $R$, interpreted as the edge relation in the graph, and the equality relation $=$. The semantics are defined in the obvious way. For example, a frame satisfies the formula $\hvarphi_{\mathrm{trans}}:=\forall xyz (x R y\wedge y R z\implies x R z)$ if and only if it is transitive. We use $\hat\ $ to denote first-order formulas, e.g., $\hvarphi$ is a first-order formula, while $\phi$ is a modal formula. We say that a frame is a $\hpsi$-frame if it satisfies the first-order formula $\hpsi$. We say that a model $M$ is a $\hpsi$-model (which we write as $M\models\hpsi$) if $M$ is based on a $\hpsi$-frame. The \emph{basic frame language} is the frame language without equality.

Following notation introduced in~\cite{HemaspaandraSchnoor-Elementary-STACS-2008}, for a first-order formula $\hpsi$, $\K{\hpsi}$ is the logic in which a modal formula $\phi$ is satisfiable if and only if there is a $\hpsi$-model $M$ and a world $w\in M$ such that $M,w\models\phi$. Such logics are called elementary modal logics. For a formula $\hpsi$ over the frame language, we consider the following problem:

\decisionproblem{$\sat{\K{\hpsi}}$}{A modal formula $\phi$}{Is there a $\hpsi$-model $M$ and a world $w\in M$ with $M,w\models\phi$?}

As an example, the problem $\sat{\K{\hvarphi_{\mathrm{trans}}}}$ is the problem of deciding if a given modal formula can be satisfied in a transitive frame, and is the same as the satisfiability problem for the logic $\logicname{K4}$. We say that a modal formula is \emph{globally} satisfied in a model if it holds at every world in the model.

The main result of this paper is that there exists a universal first-order formula $\formulaFinal$ over the basic frame language such that $\sat{\K\formulaFinal}$ is $\CORE$-complete.

\section{Relation with previous work}\label{sect:previous}

The result from the literature closest to ours is Theorem 3.2 from~\cite{Hemaspaandra-PRICE-OF-UNIVERSALITY-NDJFC-1996}, which shows the coRE-hardness of \emph{global} modal satisfiability for the class of frames defined by the following first-order universal formula (we omit quantification from first-order formulas, since in this paper every variable in every appearing first-order formula is universally quantified):
\begin{eqnarray*}
((xRy_1\wedge xRy_2\wedge xRy_3) & \rightarrow & ((y_1=y_2)\vee(y_1=y_3)\vee(y_2=y_3))) \ \ \  \wedge \\
\left(\bigwedge_{1\leq i\leq 4}(x R y_i \wedge y_i R z_i)\right. & \implies & \left.\bigvee_{1\leq i<j\leq 4}(z_i=z_j) \right).
\end{eqnarray*}

Our result strengthens the above result in several ways: We prove undecidability already for the (local) satisfiability problem, i.e., given a modal formula $\varphi$, decide whether there is a model and \emph{some} world in it that satisfies $\varphi$. Satisfiability is often much easier than global satisfiability: In the above example from~\cite{Hemaspaandra-PRICE-OF-UNIVERSALITY-NDJFC-1996}, (local) satisfiability is in \NP, while global satisfiability is undecidable. Intuitively, proving undecidability of the global satisfiability problem by encoding the grid is easier than undecidability for the local problem due to the following: A key property of the grid that needs to be enforced in the construction is that every world in a model has a specific set of reachable worlds (each world has exactly two successors reachable in one step, and three worlds that it can reach in two steps). Clearly, enforcing the existence of successors is impossible using only universally quantified formulas. With a modal formula however, requiring at least two direct successors can be enforced very easily by using a formula $\Diamond u\wedge\Diamond\overline{u}$. In the global satisfiability problem, one can require this formula to globally hold in the model, and hence enforce the existence of the grid structure relatively easily. In the (local) satisfiability problem however, there is no way to require this formula to globally hold in the model. This is the main reason why, at least in the context of classes of frames defined by universal first-order formulas, the global satisfiability problem allows us to express ``positive'' conditions (such as the existence of at least two successors) in a much more easy way than the satisfiability problem.

In our proof, we employ some of the techniques used to obtain the above-mentioned result in~\cite{Hemaspaandra-PRICE-OF-UNIVERSALITY-NDJFC-1996}. In particular, the first step in our proof is to establish undecidability of the global satisfiability problem over a class of frames that is similar to the one defined by the above formula. However, this step does much more than simply reproving the result from~\cite{Hemaspaandra-PRICE-OF-UNIVERSALITY-NDJFC-1996}: The class of frames that we construct here is considerably less restrictive than the one defined above. One reason for this is that with the more restricted language we have available (the basic frame language, which does not contain equality), we cannot restrict our frames as strongly as with the above formula from~\cite{Hemaspaandra-PRICE-OF-UNIVERSALITY-NDJFC-1996}. More importantly however, we need to construct a class which then can serve as the basis for proving hardness of (local) satisfiability. Therefore, we construct a class of frames that is tailor-made for being able to prove our main result later. A key issue here is that of reflexivity: In contrast to the grid model, we establish undecidability for a class of frames that includes reflexive frames---in fact, the reflexive frames are those which will later allow us to reduce the global satisfiability problem to the (local) satisfiability problem. In particular, directly reducing global satisfiability for the class of frames defined by the formula above to (local) satisfiability for a class of frames defined by a universal first-order formula over the basic frame language does not appear to be easier than our approach.

Technically, in addition to our first-order formulas being more complex than those used in~\cite{Hemaspaandra-PRICE-OF-UNIVERSALITY-NDJFC-1996}, a main difference is the use of what we call \emph{abstraction} of a model: Essentially, our first-order formulas enforce the relevant conditions of the grid not in a model $M$ itself, but in a model $\nicefrac M\sim$ obtained from $M$ by compressing cliques of worlds into a single world.

\section{Main Result}\label{sect:core}
 
We show that modal satisfiability is undecidable already for a formula over the basic frame language in which every appearing variable is universally quantified:

\begin{theorem}
 There exists a universal first-order formula $\formulaFinal$ over the basic frame language such that $\sat{\K\formulaFinal}$ is $\CORE$-complete.
\end{theorem}

It is well known and easy to show that modal satisfiability for a class of frames defined by a first-order formula is in $\CORE$ since this problem can be phrased as the negation of a first-order implication (see, e.g.,~\cite[Lemma 6.32]{blrive01}). It therefore remains to construct $\formulaFinal$ such that $\sat{\K\formulaFinal}$ is $\CORE$-hard.

We also mention that as an immediate corollary of the above, we obtain the following result: The \emph{uniform} version of the satisfiability problem, in which both the modal and the first-order formula are given in the input, is $\CORE$-complete.

\section{Proof of the main result}\label{sect:proof}

\subsection{Global satisfiability in the grid model}\label{sect:grid}

We prove \CORE-hardness by a reduction from the global grid satisfiability problem: The \emph{grid frame} has world set $\mathbb N\times\mathbb N$, and the accessibility relation $R=\set{((i,j),(i+1,j)),((i,j),(i+1,j))\mid i,j\in\mathbb N}$. A \emph{grid model} is a model based on the grid frame. The \emph{global grid satisfiability problem} is the following:
 
\decisionproblem{\ggrid}{A modal formula $\psi$}{Is there a grid-model $M$ that globally satisfies $\psi$?}

In \cite{Hemaspaandra-PRICE-OF-UNIVERSALITY-NDJFC-1996}, the following theorem was proven:

\begin{theorem}\label{theorem:ggrid is core hard}
 \ggrid\ is \CORE-hard, even restricted to inputs $\psi$ with $\md\psi\leq 1$.
\end{theorem}

Our proof will force models to behave ``essentially grid-like.'' However, it is immediately clear that universal first-order formulas are not expressive enough to accomplish this: It is easy to see that if $F$ is a frame, and $F'$ is a subframe of $F$ (i.e., an induced subgraph), then $F'$ satisfies every universal first-order formula over the frame language that $F$ satisfies. More precisely, universal first-order formulas \emph{with equality} can express exactly those graph properties that can be characterized with a finite forbidden subgraph. Hence with equality in our first-order language, we can express properties like ``each world has at most two direct successors,'' which is a characteristic property of the grid frame, using the formula $(xRy_1\wedge xRy_2\wedge xRy_3)\rightarrow((y_1=y_2)\vee(y_1=y_3)\vee(y_2=y_3))$.  Note that, as explained in Section~\ref{sect:previous}, this formula was used in the proof
of~\cite[Theorem 3.2]{Hemaspaandra-PRICE-OF-UNIVERSALITY-NDJFC-1996}.

As mentioned previously, requiring ``positive'' conditions, e.g., that every world in a model has direct successors with certain properties, is reasonably easy in the global satisfiability setting, but considerably more difficult in the local setting. In addition, even the above property of having at most two distinct successors cannot be expressed by a universal first-order formula in the \emph{basic} frame language, i.e., the frame language without equality: It is easy to see that no such formula can distinguish the singleton with two successors from the singleton with three successors. 

Therefore, in order to simulate the grid in our modal models, we need to exploit the interplay between the first-order and the modal aspect of the satisfiability problem. The main ingredients to our proof are the following:

\begin{enumerate}
 \item In order to circumvent having to express equality in the above formula, we employ the following construction: We use a first-order formula similar to the one above, except that instead of requiring $y_i=y_j$, we demand that there is a symmetric edge between these worlds with the first-order formula $(y_iRy_j)\wedge(y_jRy_i)$. With some additional requirements, we ensure that the relation $x\sim y$ is an equivalence relation, where $x\sim y$ holds if there is a symmetric edge between $x$ and $y$. We then ``abstract'' modal models to their $\sim$-equivalence classes. This construction allows us to ``simulate'' equality in the basic frame language. On the abstraction we can therefore express all ``forbidden subgraph'' properties of the grid, since as mentioned above, these properties can be expressed using universal first-order formulas with equality. 

 \item In order to ensure that the above-mentioned abstraction is sound, we use modal formulas to ensure that the relation $\sim$ is not only an equivalence relation, but also respects propositional assignments: For a relevant subset of variables, $\sim$-equivalent worlds will share the same valuation. This allows us to regard equivalence classes as single worlds in the abstracted model.
 \item For the structure of the model, it remains to express the ``positive'' properties of the grid, for example that every world indeed has two distinct successors. While the existence of successor worlds can be easily required with the modal operator $\Diamond$ (which in this step of the proof we can essentially use ``globally'' as we are still only dealing with the global satisfiability problem), we need to ensure that there exist successors in a different equivalence class---i.e., successors that are still present in the abstraction. We express this requirement using subtle interplay between first-order and modal requirements.
 \item Finally, the most important issue is to express the global requirement of \ggrid: The formula $\psi$ is required to hold at every world of the grid, whereas modal satisfiability is an existential property. To solve this, we force the existence of a ``universal'' world, i.e., a world connected to every other world. In this world, the global requirement ``$\psi$ must hold everywhere'' can be simulated with the (local) $\Box$ operator. In this part of the proof, we crucially rely on features of the class of models considered in the first part, which allow us to perform this construction.
\end{enumerate}

\subsection{Expressing the grid: Universal first-order aspects}

We follow the proof strategy outlined above, and start by defining a set of universal first-order formulas that force the ``abstraction'' of a model to obey some essential properties of the grid. The formulas are intuitively understood best when thinking of symmetric edges as ``equality.'' As mentioned above, we use $x\sim y$ to express that in a frame, there is an edge from $x$ to $y$, and one from $y$ to $x$ (the frame will always be clear from the context). We also use $x\sim y$ as an abbreviation for $(xRy \wedge yRx)$ in formulas.

\begin{definition}
Let $\formulaOneStep$ be the universal first-order formula 
$$(x R y_1\wedge x R y_2\wedge x R y_3)\implies\left(\bigvee_{1\leq i\leq 3}(x\sim y_i)\vee\bigvee_{1\leq i<j\leq 3}(y_i\sim y_j)\right).$$ 
\end{definition}

This formula corresponds to the property mentioned above: In the grid, each world has only two distinct successors. When reading $\sim$ as equality, the formula exactly captures this requirement. We use an analogous approach to state another important feature of the grid: While each world has two distinct successors, say $y_1$ and $y_2$, and each of these again has two distinct successors, say $z^1_1$, $z^1_2$, $z^2_1$ and $z^2_2$, in the grid each world can reach only three distinct worlds in two steps. Hence, two of the $z^i_j$ must coincide. This is expressed by the following formula:

\begin{definition}
 Let $\formulaTwoStep$ be the universal first-order formula 
 $$\left(\bigwedge_{1\leq i\leq 4}(x R y_i \wedge y_i R z_i)\right)\implies\left(\bigvee_{1\leq i\leq 4}(x\sim y_i)\vee\bigvee_{1\leq i\leq 4}(y_i\sim z_i)\vee\bigvee_{1\leq i<j\leq 4}(z_i\sim z_j)\right).$$
\end{definition}

When reading $\sim$ as equality, this formula states that if $z_1,z_2,z_3,z_4$ are worlds reachable from $x$ via intermediate worlds $y_i$ (which are different from both $x$ and all $z_i$), then two of the $z_i$ must coincide. Note that the formulas introduced up to now closely mirror the formula in~\cite{Hemaspaandra-PRICE-OF-UNIVERSALITY-NDJFC-1996}. The major differences and additions in the construction follow now. The third formula we use ensures the relation $\sim$ mentioned above is an equivalence relation:

\begin{definition}
 Let $\formulaEQ$ be the first-order formula
 $$\big((x\sim y)\wedge y R z)\implies x R z\big)\wedge\big((x\sim y)\wedge z R y)\implies z R x\big).$$
\end{definition}

This formula ensures that $\sim$-``equivalent'' worlds have the exact same in- and outgoing edges, hence it forces $\sim$ to be an equivalence relation in reflexive models. The conjunction of the formulas above combines the first-order aspects of the grid that we can force with universal formulas over the basic frame language:

\begin{definition}
 Let $\formulaEQOneTwoStep:=\formulaOneStep\wedge\formulaTwoStep\wedge\formulaEQ$.
\end{definition}

\subsection{Properties of abstracted frames}

We now formally define abstractions of frames, which as mentioned before are obtained by compressing $\sim$-equivalence classes into a single world.

\begin{definition}
For a reflexive frame $F=(W,R)$, where $F\models\formulaEQ$, we define the \emph{abstraction of $F$}, denoted with $\nicefrac{F}{\sim}$, to be the frame $\left(\nicefrac{W}{\sim},\nicefrac{R}{\sim}\right),$ where $\nicefrac{W}{\sim}$ is the set of $\sim$-equivalence classes of $W$, and $[w]\nicefrac{R}{\sim}[w']$ if and only if $w R w'$.  
\end{definition}

The relation $\nicefrac{R}{\sim}$ above is well-defined, since if $w\sim w'$ and $\hat w\sim\hat{w}',$ then $w R \hat w$ implies $w' R \hat w,$ and this implies $w' R \hat{w}'$ (by $\formulaEQ$). We now show that for a frame that satisfies the formula $\formulaEQOneTwoStep$, its abstraction has two key properties of the grid frame, and additionally is reflexive---reflexivity does not actually help in establishing the undecidability result for global satisfiability, but as commented before will make it easier for us to later move to the (local) satisfiability problem.

\begin{lemma}\label{lemma:structucal properties of frames satisfying basic formula}
 Let $F=(W,R)$ be a reflexive $\formulaEQOneTwoStep$-frame. Then $\nicefrac F\sim$ satisfies the following properties:
\begin{enumerate}
 \item $\nicefrac F\sim$ is reflexive,
 \item each world $[w]$ in $\nicefrac F\sim$ has at most two direct successors different from $[w]$,
 \item for each world $[w]$ in $\nicefrac F\sim$, there  are at most three
worlds that can be reached on a path from $[w]$ of length two that does not use any reflexive edge.
\end{enumerate}
\end{lemma}

\begin{proof}
 \begin{enumerate}
  \item Let $[w]$ be an equivalence class. Since $R$ is reflexive, we know that $w R w$, hence due to the definition of $\nicefrac{(W,R)}{\sim}$, we have that $[w]\nicefrac R\sim[w]$.
  \item Assume that some equivalence class $[w]$ has direct successors $[w_1]$, $[w_2]$, and $[w_3]$ in $\nicefrac{(W,R)}{\sim}$, where all of these four classes are distinct. Obviously, for $i\in\set{1,2,3}$, we have that $([w_i],[w])\notin\nicefrac{R}{\sim}$, since otherwise $[w_i]=[w]$ would follow due to the definition of $\sim$. By the properties of $\formulaOneStep,$ it then follows that there are $i,j$ such that $1 \leq i < j \leq 3$ and $w_i R w_j$ and $w_j R w_i,$ thus $w_i\sim w_j,$ and hence $[w_i]=[w_j],$ a contradiction.
  \item This follows very similarly: Assume that there is a class $[x]$ with direct successors $[y_1]$, $[y_2]$, $[y_3]$, and $[y_4]$, and classes $[z_1]$, $[z_2]$, $[z_3]$, $[z_4]$ such that for all $i,j$ we have $[y_i]\neq[x]$, $[y_i]\nicefrac R \sim [z_i]$, $[y_i]\neq[z_i]$, and for $i\neq j$ we have $[z_i]\neq[z_j]$. Since $[x] \nicefrac R\sim [y_i]$, there is an edge $xRy_i$ for all $i$. If there was an edge $y_i R x$ for some $i$, then $y_i\sim x$ and thus $[x]=[y_i]$ would follow, hence there is no such edge. Similarly, since $y_i \nicefrac R\sim z_i$ but $[y_i]\neq[z_i]$, there is no edge $z_i R y_i$. Therefore, since $(W,R)\models\formulaTwoStep$, we know that there exist $i\neq j$ with $z_i Rz_j$ and $z_j Rz_i$, i.e., $z_i\sim z_j$. This implies $[z_i]=[z_j]$, a contradiction.
 \end{enumerate}
\end{proof}

We therefore know that the abstracted frame satisfies the main ``forbidden subgraph'' properties of the grid frame---as mentioned before, we cannot hope to enforce other properties of the grid frame using only universal first-order formulas. For ensuring the remaining properties, we therefore use modal formulas and propositional variables. There are two main differences between the grid frame and the class of abstractions of $\formulaEQOneTwoStep$-frames: First, worlds in abstractions of $\formulaEQOneTwoStep$-frames that correspond to more than one world in the original frame are always reflexive. Second, the class of $\formulaEQOneTwoStep$-frames (and thus the class of its abstractions) is subframe-closed. Essentially, abstractions of $\formulaEQOneTwoStep$-frames can be seen as subframes of the reflexive closure of the grid frame.

\subsection{Expressing the grid: Modal aspects}

Above, we have expressed the universal first-order properties of the grid frame in the formula $\formulaEQOneTwoStep$. We now consider the second part of our abstraction process: the valuation of the propositional variables. We will use the interplay between modal formulas and the frame properties ensured with the results in the previous section to address the following issues:

\begin{itemize}
 \item We need to ensure that the abstraction is ``sound,'' i.e., leaves crucial modal properties of the models invariant. The main issue is that we need to force all ``relevant'' variables to have the same value in all $\sim$-equivalent worlds. This leads to a well-defined propositional assignment for the abstracted frame, and moreover ensures that truth of ``relevant'' modal formulas remains invariant when moving to the abstraction of a model.
 \item As mentioned above, the abstractions constructed in the previous section are necessarily reflexive. To simulate the (non-reflexive) grid, we will replace the $\Box$-operator with a construction that ensures that $\Box\psi$ only requires the formula $\psi$ to be true in the successors $w'\neq w$ of a world $w$.
 \item We have to enforce the ``positive'' properties of the grid frame, i.e., that every world in fact does have two distinct successors.
\end{itemize}

Our abstraction does not take into account all variables, but only a subset denoted with $P$---this will contain all variables appearing in the input formula $\psi$ for the $\ggrid$ problem that we will reduce from.

\begin{definition}
Let $P$ be a set of propositional variables, and let $M = (W,R,\pi)$ be a
reflexive $\formulaEQ$-model. We define the model $\nicefrac{M}{\sim}$ as

$$\nicefrac{M}{\sim}:=\left(\nicefrac{W}{\sim},\nicefrac{R}{\sim},\nicefrac{\pi}{\sim}\right),$$
where the assignment $\nicefrac{\pi}{\sim}$ makes the variable $p$ true in an equivalence class $[w]$ if and only if it is true in all elements of $[w]$, and lets all propositional variables not in $P$ be false everywhere. We call $\nicefrac M\sim$ the \emph{abstraction} of $M$.
\end{definition}

We will later only consider abstractions of models where the truth value of a variable $p\in P$ does not depend on the representative of a class $[w]$. This will be enforced to ensure that $\nicefrac{M}{\sim}$ is a sound abstraction of with respect to the propositional variables in $P$. For models that obey this restriction, the above truth assignment will then be equivalent to setting $p$ true in $[w]$ if there is some $w\in[w]$ where $p$ is true in the original model. We say that \emph{$\sim$ respects $P$ in a model $M$} if $w\sim w'$ implies that each variable $p\in P$ is true in $w$ if and only if it is true in $w'$. We often omit the model $M$ if clear from the context. 

It is easy to see that this property implies that the abstraction introduced is indeed sound, i.e., preserves truth of all modal formulas---at least as long as we only consider reflexive models that also satisfy the first-order formula $\formulaEQ$:

\begin{lemma}\label{lemma:modal formulas invariant if sim respects variables}
 Let $P$ be a set of propositional variables, let $M$ be a reflexive $\formulaEQ$-model such that $\sim$ respects $P$ in $M$. Then for all $w\in W,$ and all modal formulas $\psi$ with $\var{\psi}\subseteq P$, $M,w\models\psi\mathtext{ if and only if }\nicefrac{M}{\sim},[w]\models\psi$.
\end{lemma}

\begin{proof}
 We prove the claim by induction on the construction of $\psi.$ If $\psi$ is a propositional variable $p\in P$, then this follows from the prerequisites, since for $w\sim w',$ $p$ is true at $w$ if and only if $p$ is true at $w'.$ The induction step for propositional operators is trivial. Therefore, assume that $\psi$ is of the form $\Box\xi,$ and the claim holds for $\xi.$ Since $w R w'$ is true if and only if $[w] \nicefrac{R}{\sim} [w'],$ the following is true:

\begin{tabular}{lll}
 $M,w\models\Box\xi\ $ & iff\ \ & for all $w'$ with $w R w',$ it holds that $M,w'\models\xi$ \\
                       & iff    & for all $w'$ with $w R w'$ it holds that $\nicefrac{M}{\sim},[w']\models\xi$ \\
                       & iff    & for all $w'$ with $[w] \nicefrac{R}{\sim} [w']$ it holds that $\nicefrac{M}{\sim},[w']\models\xi$ \\
                       & iff    & $\nicefrac{M}{\sim},[w]\models\Box\xi.$
\end{tabular}
\end{proof}

We therefore obtain the following: If $M$ is a reflexive modal model that satisfies the formula $\formulaEQOneTwoStep$ and in which $\sim$ respects $P$, then its abstraction is a subframe of the reflexive closure of the grid that satisfies the same formulas $\psi$ as $M$ does (for formulas $\psi$ with $\var\psi\subseteq P$). Therefore, if we can enforce that $\sim$ respects $P$, we have can enforce the abstractions of our models to exhibit the ``forbidden subgraph''-features of the grid, without changing the set of satisfied modal formulas. To enforce that $\sim$ respects $P$, we define the following:

\begin{definition}
For a set $P$ of propositional variables, let $\formulaRespectP$ be defined as

$$
\displaystyle
\begin{array}{lcl}
\formulaRespectP & = & \bigwedge_{d=0}^7\bigwedge_{p\in P} 
((d_8 = d) \rightarrow \\
&& \ \ \Box ((d_8 = d) \vee (d_8 = (d + 2) \mod 8) \vee (d_8 = (d+3) \mod 8))\\
&& \ \ \wedge \  (p\implies\Box((d_8=d)\implies p)) \\
&& \ \ \wedge \  (\overline{p}\implies\Box((d_8=d)\implies\overline{p})))
\end{array}
$$
\end{definition}

Here $d_8$ is a variable over the natural numbers $0,\ldots,7$. This obviously can be represented by three propositional variables $d_8^a$, $d_8^b$, and $d_8^c$ as follows:

$$
\begin{array}{lll}
 d_8=0 & \leftrightarrow & (\overline{d_8^a}\wedge\overline{d_8^b}\wedge
\overline{d_8^c}), \\
 d_8=1 & \leftrightarrow & (\overline{d_8^a}\wedge\overline{d_8^b}\wedge
d_8^c), \ldots \\
 d_8=7 & \leftrightarrow & ({d_8^a}\wedge d_8^b \wedge d_8^c).
\end{array}
$$

For a model $M$ and a world $w\in M$, with $d_8(w)$ we denote the unique value $d\in\set{0,\ldots,7}$ such that $M,w\models (d_8=d)$. To increase readability, we often omit the ``$\mathtext{mod}\ 8$'' from the discussion and write $d_8 = d$ for $d_8 = d \mod 8$.

We now prove that the formula $\formulaRespectP$ works as intended, if it is \emph{globally} satisfied in a model---recall that a model $M$ globally satisfies a modal formula $\psi$ if $M,w\models\psi$ for all worlds $w$ of $M$.

\begin{lemma}\label{lemma:formula forces sim to respect P}
 Let $P$ be a set of propositional variables, let $M$ be a reflexive $\formulaEQOneTwoStep$-model that globally satisfies $\formulaRespectP$. Then $\sim$ respects $P\cup\set{d_8}$ in $M$.
\end{lemma}

\begin{proof}
 Let $w\sim w'$ in $M$, and let $d=d_8(w)$. We first prove that $M,w'\models(d_8=d)$, i.e., that the values of $d_8$ are identical in $\sim$-equivalent worlds. Suppose that this is not the case. Since $M$ globally satisfies $\formulaRespectP$, and there is an edge $w R w'$, we know that $M,w'\models(d_8 = (d+2))$ or $M,w'\models(d_8=(d+3))$. Since there also is an edge $w' R w$, it also follows from $\formulaRespectP$ that $M,w\models(d_8 \in \{d+2+2, d+2+3, d+3+2, d+3+3\})$, which is a contradiction, since $M,w\models(d_8=d)$. We therefore know that $M,w'\models(d_8=d)$ as claimed. Now let $p$ be a variable from $P$. Since $w R w'$ and $M,w\models\formulaRespectP$ from $M,w\models p$ it follows that $M,w'\models p,$ and $M,w\models\overline{p}$ implies $M,w'\models\overline p$. Therefore, $\sim$ respects $P$ as required.
\end{proof}

The above lemma together with the earlier Lemma~\ref{lemma:modal formulas invariant if sim respects variables} and the fact that $\var{\formulaRespectP}=\set{d_8}$ imply the following:

\begin{corollary}\label{corollary:if respect P satisfied than formulas lift to abstraction}
  Let $M$ be a reflexive $\formulaEQ$-model such that $M$ globally satisfies $\formulaRespectP$. Then for all formulas $\psi$ with $\var\psi\subseteq P\cup\set{d_8}$, and all worlds $w\in M$, we have that 
  $$M,w\models\psi\mathtext{ if and only if }\nicefrac{M}{\sim},[w]\models\psi.$$
  In particular, $\nicefrac M\sim$ globally satisfies $\formulaRespectP$.
\end{corollary}

The variable $d_8$ also allows to express ``positive'' conditions of the grid, i.e., the existence of certain successor worlds, and to reason about direct successors of a world $w$ that are not $\sim$-equivalent to $w$ itself. In the sequel, we call successors like this \emph{non-symmetric successors of $w$}. Similarly, a \emph{non-reflexive successor} of $w$ is a direct successor $w'$ of $w$ with $w\neq w'$. This does not mean that $w'$ is an irreflexive world, but that it is one reachable from $w$ with an edge other than the reflexive loop. The following formula now ensures that worlds that are connected with a direct edge, but are not $\sim$-equivalent, have different values for $d_8$. 

\begin{definition}
 Let $\formulaExistsNontrivialSuccessor$ be defined as follows
$$\formulaExistsNontrivialSuccessor=\displaystyle\bigwedge_{d = 0}^7 \left( (d_8 = d) \implies \Diamond ( d_8=d + 2) \wedge \Diamond(d_8 = d + 3)\right) $$
\end{definition}

This formula expresses that a world $w$ with $d_8(w) = d$ has direct successors with $d_8 = d + 2$ and $d_8 = d + 3$. We will later identify the ``$+2$''/``$+3$''-successor with the ``upper''/``right'' neighbor in the grid. If additionally the model globally satisfies $\formulaRespectP$, then from neither of these successors, the world $w$ is reachable in one step: From the definition of $\formulaRespectP$, it follows that all direct successors of the two ones whose existence is forced by $\formulaExistsNontrivialSuccessor$ must have $d_8$-values out of the set $\set{d+2,d+4,d+5,d+3,d+5,d+6}$, none of which applies to the world $w$ itself. More generally, every world $w'\neq w$ reachable from $w$ with at most two steps has a different $d_8$-value than $w$.

In addition, the successors with $d_8$-values of $d+2$ and $d+3$ cannot be connected with a direct edge in models satisfying $\formulaRespectP$. Hence, each world has two successors such that all of the three involved worlds are $\sim$-inequivalent---in the abstraction, these worlds will give rise to three different equivalence classes.

It follows from the above that in $\formulaEQOneTwoStep$-models globally satisfying the formula $\formulaRespectP$, the formula $\bigwedge_{d=0}^7 ((d_8=d) \implies\Box ((d_8\neq d)\rightarrow\psi))$ is true in a world $w$ if and only if $\psi$ is true in all non-symmetric, non-reflexive successors of $w$.

\subsection{\CORE-hardness of global satisfiability}

We now show \CORE-hardness for the global satisfiability problem on reflexive frames that satisfy $\formulaEQOneTwoStep$. In itself, this result is not stronger than what was already established in~\cite{Hemaspaandra-PRICE-OF-UNIVERSALITY-NDJFC-1996}, except for the fact that our formula only uses the basic frame language, i.e., does not use equality. However, the real benefit of the result in this section will become apparent in the next section: The class that we define here allows us to easily reduce global to (local) satisfiability.

In Lemma~\ref{lemma:formula forces sim to respect P}, we have seen that if we can ensure that the formula $\formulaRespectP$ is \emph{globally} satisfied in a reflexive $\formulaEQOneTwoStep$-model, then $\sim$ respects $P$ and in this case Lemma~\ref{lemma:modal formulas invariant if sim respects variables} tells us that our abstraction is sound, i.e., preserves truth values of modal formulas. Since the formulas $\formulaRespectP$ and $\formulaExistsNontrivialSuccessor$ allow us to ensure that $\sim$ respects $P$ and that every world has the two distinct successors as in the grid frame, we therefore can use the construction from the previous section to prove \CORE-hardness in the case that we are able to enforce $\formulaRespectP\wedge\formulaExistsNontrivialSuccessor$ globally. 

Recall that $\ggrid$ remains $\CORE$-hard when the input is restricted to formulas $\psi$ with $\md\psi\leq 1$. We therefore only consider such inputs for $\ggrid$ from now on, and define our reduction as follows:

\begin{definition}
 Let $\psi$ be an input for $\ggrid$ with $\md\psi\leq 1$. Let $P=\var\psi$, and let $g(\psi)$ be defined inductively as follows:
 \begin{itemize}
  \item If $\psi$ is a variable $p,$ then $g(\psi)=p$, 
  \item $g(\neg\psi)=\neg g(\psi)$,
  \item $g(\psi\wedge\xi)=g(\psi)\wedge g(\xi)$,
  \item $g(\Box\psi)=
\bigwedge_{d=0}^7 ((d_8=d)
\implies\Box ((d_8\neq d)\rightarrow g(\psi)))$.
 \end{itemize}
 The reduction $f$ is now defined as $f(\psi)=g(\psi)\wedge\formulaRespectP\wedge\formulaExistsNontrivialSuccessor$.
\end{definition}

The only non-obvious part of the definition is the handling of the $\Box$-operator. As argued above, the translation of $\Box\psi$ requires $\psi$ to be true in all non-reflexive, non-symmetric successor worlds of the current one. This will be crucial when we consider abstractions of models: The non-symmetric successors of a world $w$ in a model $M$ directly correspond to the non-reflexive successors of the class $[w]$ in the model $\nicefrac M\sim$. We now prove that the reduction is correct (see Appendix).

\begin{theorem}\label{theorem:main reduction correctness}
 Let $\psi$ be an instance of $\ggrid$ with $\md\psi\leq 1$. Then $\psi$ is a positive instance of $\ggrid$ if and only if $f(\psi)$ is globally satisfiable on a reflexive $\formulaEQOneTwoStep$-model.
\end{theorem}

\begin{proof}
In the following, let $P$ again be the set of variables appearing in $\psi$. First assume that $\psi$ is a positive instance of $\ggrid$, i.e., there is a grid model  $M=(\mathbb{N}\times\mathbb{N},R,\pi)$ such that $M,(i,j)\models\psi$ for all $i,j\in\mathbb N$. We define the model $\hat M$ as the one obtained from $M$ as follows: $\hat M=(\mathbb{N}\times\mathbb{N},\hat R,\hat\pi),$ where

\begin{itemize}
 \item $\hat R=R\cup\set{((i,j),(i,j))\ \vert\ i,j\in\mathbb N},$ i.e., $\hat R$ is the reflexive closure of $R,$
 \item $\hat\pi$ agrees with $\pi$ on $\var{\psi}$, and $M,(i,j)\models(d_8=(3i+2j))$.
\end{itemize}

Again, $d_8$ can easily be expressed using three propositional variables. It is immediate that $\hat M$ is reflexive and satisfies $\formulaEQOneTwoStep$ (note that in $\hat M$, we have that $w\sim w'$ if and only $w=w'$). By choice of $\hat\pi$, it is also obvious that $\hat M$ globally satisfies $\models\formulaRespectP$ and $\formulaExistsNontrivialSuccessor$. 

It remains to show that $\hat M$ globally satisfies $g(\psi)$. Since $M$ globally satisfies $\psi$, it suffices to show that for all subformulas $\chi$ of $\psi$, and for all $i,j\in\mathbb N,$ it holds that $M,(i,j)\models\chi$ if and only if $\hat M,(i,j)\models g(\chi).$ We prove this by induction on the construction of $\chi$. Clearly, the only non-trivial case is when $\chi=\Box\xi$ for some $\xi$. Due to induction and since $M$ is based on a grid frame, the following holds:

$$
\begin{array}{llll}
M,(i,j)\models\Box\xi\ \ \ & \mathtext{iff} &\ M,(i+1,j)\models\xi\mathtext{ and }M,(i,j+1)\models\xi \\
& \mathtext{iff} &  \hat M,(i+1,j)\models g(\xi)\mathtext{ and }\hat M(i,j+1)\models g(\xi) \\
& \mathtext{iff} & \hat M,(i,j) \models\Box((d_8\neq d_8((i,j)))\implies g(\xi))\\
& \mathtext{iff} & \hat M,(i,j)\models g(\Box\xi).
\end{array}
$$

Hence we know that $\hat M,(0,0)$ globally satisfies  $f(\psi)$, as required.

For the other direction, assume that $f(\psi)$ is globally satisfied on a reflexive $\formulaEQOneTwoStep$-model $M=(W,R,\pi)$.  From $M$, we now obtain a grid model as follows: We first consider the abstraction of $M$, which, since $M$ is a $\formulaEQOneTwoStep$-model and $M$ globally satisfies $\formulaExistsNontrivialSuccessor$, is essentially grid-like due to Lemma~\ref{lemma:structucal properties of frames satisfying basic formula}. Since $M$ also globally satisfies $\formulaRespectP$, this abstraction is sound, i.e., still globally satisfies $g(\psi)$. We can then easily modify $\nicefrac M\sim$ to obtain a model that in fact is a grid and globally satisfies $\psi$. More formally, let $M^0\gd$ be defined as $(W^0\gd,R^0\gd,\pi^0\gd)$, where 

\begin{itemize}
 \item $W^0\gd=\nicefrac{W}{\sim},$
 \item $R^0\gd=\set{([w],[w'])\in\nicefrac{R}{\sim}\ \vert\ d_8([w])\neq d_8([w'])}$,
 \item $\pi^0=\nicefrac\pi\sim$.
\end{itemize}

The above choice of $R^0$ is well-defined, since due to Lemma~\ref{lemma:formula forces sim to respect P}, we know that the value of $d_8$ does not depend on the choice of the representative of a $\sim$-equivalence class $[w]$. By construction, since $M$ globally satisfies $\formulaRespectP$, we know that $M^0\gd$ is exactly the model $\nicefrac M\sim$ with the reflexive edges removed. Since due to Corollary~\ref{corollary:if respect P satisfied than formulas lift to abstraction}, we know that $\nicefrac M\sim$ globally satisfies $\formulaRespectP$, and satisfaction of this formula clearly is invariant under removing reflexive edges, it follows that $M^0\gd$ also globally satisfies $\formulaRespectP$.

We now prove that for all subformulas $\chi$ of $\psi$ and all worlds $w\in M$, we have that $M,w\models g(\chi)$ if and only if $M^0\gd,[w]\models\chi$. Since $M$ globally satisfies $\formulaRespectP$, we know from Corollary~\ref{corollary:if respect P satisfied than formulas lift to abstraction} that $M,w\models g(\chi)$ if and only if $\nicefrac M\sim,[w]\models g(\chi)$. (Note that $\var{g(\chi)}=\var{\chi}\cup\set{d_8}\subseteq P\cup\set{d_8}$.) It therefore suffices to prove that $\nicefrac M\sim,[w]\models g(\chi)$ iff $M^0\gd,[w]\models \chi$. We prove the claim by induction on $\chi$. Since the only difference between $M^0\gd$ and $\nicefrac M\sim$ is the set of edges between worlds, the only interesting case is when  $\chi=\Box\xi$. Due to Corollary~\ref{corollary:if respect P satisfied than formulas lift to abstraction}, $\nicefrac M\sim$ globally satisfies $\formulaRespectP$. Hence the non-reflexive successors of a world $[w]$ in $\nicefrac M\sim$ are exactly those successors $[w']$ of $[w]$ with $d_8([w'])\neq d_8([w])$ (recall that $d_8(w)$ only depends on the equivalence class $[w]$ due to Lemma~\ref{lemma:formula forces sim to respect P}).

Since $g(\chi)=\bigwedge_{d=0}^7((d_8=d)\implies\Box((d_8\neq d)\rightarrow g(\xi)))$, we have the following:

\medskip

\begin{tabular}{rl}
                       & $\nicefrac M\sim,[w]\models g(\chi)$ \\
                   iff & $\nicefrac M\sim,[w]\models\bigwedge_{d=0}^7(d_8=d)\rightarrow\Box((d_8\neq d)\rightarrow g(\xi)))$ \\
                   iff & $\nicefrac M\sim,[w']\models g(\xi)$ for all $[w']\neq[w]$ with $([w],[w'])\in\nicefrac R\sim$ \\
       iff (induction) & $M^0\gd,[w']\models\xi$ for all $[w']\neq[w]$ with $([w],[w'])\in\nicefrac R\sim$ \\
iff (def. of $M^0\gd$) & $M^0\gd,[w']\models\xi$ for all $[w']$ with $([w],[w'])\in R^0\gd$ \\
                   iff & $M^0\gd,[w]\models\Box\xi$ \\
                   iff & $M^0\gd,[w]\models\chi$. \\
\end{tabular}

\medskip

This completes the proof of the above claim. Since $M$ globally satisfies $g(\psi)$, this implies that $M^0\gd$ globally satisfies $\psi$. We now construct from $M^0\gd$ a grid-model $M\gd$ that still globally satisfies $\psi$ as required. Recall that $M$ also satisfies the first-order formula $\formulaEQOneTwoStep$, and that $M^0\gd$ is obtained form $\nicefrac M\sim$ by removing edges. Therefore, Lemma~\ref{lemma:structucal properties of frames satisfying basic formula} implies that from each world in $M^0\gd$ there are at most two worlds reachable in one step, and at most three worlds reachable in two steps. Further, since $M$ globally satisfies $\formulaExistsNontrivialSuccessor$, we know that each world in $\nicefrac M\sim$ has at least two distinct successors. Since $M^0\gd$ is obtained from $M$ by removing reflexive edges, $M^0\gd$ also has this property.

The values of $d_8$ in the individual worlds induces an ordering on the direct successors of a world $[w]$ in $M^0\gd$. First, recall that due to Lemma~\ref{lemma:formula forces sim to respect P}, the values of $d_8$ depend only on the equivalence of the worlds, hence we may use $d_8[w]$ to denote the $d_8$-value of \emph{all} worlds in the equivalence class $[w]$. Due to the properties of the abstraction, $[w]$ has a unique direct successor world $[w]^\uparrow$ with $d_8([w^\uparrow])=d_8([w])+2$, and a unique direct successor world $[w]^\rightarrow$ with $d_8([w^\rightarrow])=d_8([w])+3$. In addition, since due to Lemma~\ref{lemma:structucal properties of frames satisfying basic formula}, $[w]$ can only reach three worlds on a path of length two that does not use any reflexive edges, we know that $\card{\set{[w]^{\uparrow\rightarrow},[w]^{\uparrow\uparrow},[w]^{\rightarrow\rightarrow},[w]^{\rightarrow\uparrow}}}=3$, hence two of these worlds must be the same. Due to the distribution of the $d_8$-values, it follows that $[w]^{\uparrow\rightarrow}=[w]^{\rightarrow\uparrow}$, since these are the only two of the mentioned worlds that share the same $d_8$-value, namely $d_8([w])+5$.

Thus $M^0\gd$ can be written as a grid model using standard unfolding techniques: We define the grid model $M\gd$ as follows:

\begin{itemize}
 \item the world $(0,0)$ is a copy of some world $[w]$ of $M^0\gd$.
 \item if $(i,j)$ is a copy of the world $[w]_{i,j}$, then let the worlds $(i+1,j)$ and $(i,j+1)$ be copies of the worlds $[w]_{i,j}^\rightarrow$ and $[w]_{i,j}^\uparrow$ of $M^0\gd$, and ensure that for $(i+1,j+1)$, the same copy of $[w]^{\rightarrow\uparrow}=[w]^{\uparrow\rightarrow}$ is used.
\end{itemize}

It is clear that the set of modally satisfied formulas does not change in the step from $M^0\gd$ to $M\gd$, and that $M\gd$ is indeed a grid model. Therefore, $M\gd$ is a grid model that globally satisfies $\psi$ as required.
\end{proof}

We mention that one can easily to use the first-order formula to force the models to be reflexive, using the clause $xRx$. However, to be able to prove undecidability for satisfiability instead of global satisfiability, it is crucial to leave open the possibility of non-reflexive worlds, as we will see in the next section.

\subsection{Removing globalness}

The construction in the preceding section showed hardness for global satisfiability for reflexive $\formulaEQOneTwoStep$ frames. To obtain our $\CORE$-hardness result for (local) satisfiability, we now express this global quantification with only the first-order frame language and the modal language.

The main idea of the proof is the following: The construction forces the existence of a ``universal'' world $w_u$, i.e., a world that has an outgoing edge to every other world in the model. Since this is an ``existential'' and not a ``forbidden subgraph'' property, we cannot express this as a universal first-order formula directly. We therefore use the following construction: We require that for every pair of a world $w_u$ that is not reflexive, and every world $w$ that has an incoming edge, there is an edge from $w_u$ to $w$. This ensures that if the model contains an irreflexive world $w_u$, then $w_u$ is universal at least with respect to worlds that can be reached from any other world at all. In particular, $w_u$ is universal with respect to the submodel rooted at $w_u$. Additionally, we require that any world that has an incoming edge is reflexive. We therefore have established that if there is a world $w_u$ that is irreflexive, then every world reachable from $w_u$ in any number of steps is connected to $w_u$ directly, and every such world is reflexive. 

These conditions can be enforced with the following formula:

\begin{definition}
Let $\formulaUniversalWorld$ be the universal first-order formula 
$$(x R y\implies y R y)\wedge (\overline{w_u R w_u} \implies (x R y \implies w_u R y)).$$ 
\end{definition}

The existence of an irreflexive world $w_u$ can easily be enforced with the modal formula $u\wedge\Box\overline u$, where $u$ is a new variable. We then enforce the formula $\formulaEQOneTwoStep$ constructed in the previous section only on reflexive worlds, and can thus identify the ``reflexive part'' of a model with a model of the type as considered in the previous section. In particular, we know that global satisfiability of a formula of the form $f(\psi)$ on the ``reflexive part'' of our models is \CORE-hard, where $f$ is the function used in the reduction from Theorem~\ref{theorem:main reduction correctness}. We then use the universal world $w_u$ to express the global satisfiability problem with a single $\Box$-operator.

We therefore obtain the following theorem:

\begin{theorem}\label{theorem:core hardness for grid}
 There exists a universal first-order formula $\formulaFinal$ over the basic frame language such that $\sat{\K\formulaFinal}$ is $\CORE$-hard.
\end{theorem}

\begin{proof}
 From Theorems~\ref{theorem:ggrid is core hard} and~\ref{theorem:main reduction correctness}, we know that the global satisfiability for formulas of the form $f(\psi)$ on reflexive $\formulaEQOneTwoStep$-models is \CORE-hard.

In order to ``simulate'' global satisfiability,
we want to add a ``universal world'' to our models, i.e., a world that is connected to every world in the model except itself. The reflexivity of worlds will be used to distinguish between the universal world and other worlds in the model. As mentioned above, the effect of the formula $\formulaUniversalWorld$ is that any world $w_u$ which is non-reflexive does not have a predecessor, and is connected in one step to every world that that does have one. In particular, $w_u$ is a ``universal world'' with respect to the submodel generated by $w_u$.

Let $\formulaEQOneTwoStep \equiv \forall x_1 \cdots \forall x_k \formulaEQOneTwoStepKernel(x_1, \ldots, x_k)$, where $\formulaEQOneTwoStepKernel$ is quantifier-free. To enforce $\formulaEQOneTwoStep$, but only on the submodel containing all reflexive worlds, we define

$$\formulaEQOneTwoStepRefl:=(x_1 R x_1 \wedge \dots \wedge x_k R x_k )\implies\formulaEQOneTwoStepKernel(x_1,\dots,x_k).$$

We can now define the complete universal first-order formula $\formulaFinal$ such that $\sat{\K\formulaFinal}$ is \CORE-hard as follows:

$$\formulaFinal:=\formulaUniversalWorld\wedge\formulaEQOneTwoStepRefl.$$

We show that a modal formula $\psi$ is globally satisfiable on a reflexive $\formulaEQOneTwoStep$-frame if and only if $u\wedge\Box\overline{u}\wedge\Box\psi$ is satisfiable on a $\formulaFinal$-model. Since deciding whether the former holds is \CORE-hard due to Theorems~\ref{theorem:ggrid is core hard} and~\ref{theorem:main reduction correctness}, this proves the theorem.

First assume that $M=(W,R,\pi)$ is a reflexive $\formulaEQOneTwoStep$-model and that $M$ globally satisfies $\psi$. We define a model $\hat M=(\hat W,\hat R,\hat\pi)$ such that

\begin{itemize}
 \item $\hat W=W\cup\set{w_u},$ where $w_u$ is a new world,
 \item $\hat R=R\cup\set{(w_u,w)\ \vert\ w\in W},$
 \item $\hat\pi=\pi,$ except that the new variable $u$ is true at $w_u$ and
 false at every world other than $w_u$, and the truth values of the remaining propositional variables at $w_u$ are arbitrary.
\end{itemize}

It is obvious from the construction that $\hat M$ is based on a $\formulaFinal$-frame and that $\hat M,w_u\models u\wedge\Box\overline{u}\wedge\Box\psi$.

For the converse, let $M=(W,R,\pi)$ be a model based on a $\formulaFinal$-frame, and let $w_u\in W$ be such that $M,w_u\models u\wedge\Box\overline{u}\wedge\Box\psi$. Note that since $M,w_u\models u\wedge \Box\overline u$, $w_u$ clearly is irreflexive ($u$ holds in $w_u$, but not in any successor of $w_u$, hence $w_u$ cannot be one of these successors). Now let $\hat M=(\hat W,\hat R,\hat\pi)$ be defined as follows:

\begin{itemize}
 \item $\hat W=\set{w\ \vert\ (w_u,w)\in R}$,
 \item $\hat R=R\cap (\hat W\times\hat W)$,
 \item $\hat\pi=\pi$ restricted to $\hat W$.
\end{itemize}

Note that $\hat M$ is reflexive by the construction of $\formulaFinal$, and hence, since $M\models\formulaEQOneTwoStepRefl$, and satisfaction of universal first-order formulas is invariant under deleting worlds, it follows that $\hat M$ is based on a reflexive $\formulaEQOneTwoStep$-frame. It remains to prove that $\hat M,w\models\psi$ for all $w\in\hat W$. Since $M,w_u\models\Box\psi$, it follows that $M,w\models\psi$ for all $w\in\hat W$. It therefore it suffices to show that for all subformulas $\chi$ of $\psi$ and for all $w\in\hat W$, that $M,w\models\chi$ if and only if $\hat M,w\models\chi$. We show the claim by induction on the construction of $\chi$, and the only non-trivial case is when $\chi=\Box\xi$ for a modal formula $\xi$.

First assume that $M,w\models\Box\xi$, and let $w'$ be a successor of $w$ in $\hat M$. Since $\hat R\subseteq R$, it follows that $w'$ is also a successor of $w$ in $M$. Since $M,w\models\Box\xi$, this implies that $M,w'\models\xi$. Due to induction, it follows that $\hat M,w'\models\xi$ as required.

For the converse, assume that $M,w\nmodels\Box\xi$. Then there is some world $w'\in W$ such that hat $w R w'$ and $M,w'\nmodels\xi$. Since $wRw'$, and the world $w_u$ is irreflexive in $M$ due to the above, and $M\models\formulaUniversalWorld$, it follows that $w_u R w'$ is an edge in $M$. Therefore, by definition it holds that $w'\in\hat W$, and due to the definition of $\hat R$ we have that $w \hat R w'$ in $\hat M$. Since due to induction we know that $\hat M,w'\nmodels\xi$, it follows that $\hat M,w\nmodels\Box\xi$ as required. This completes the induction and therefore the proof of the theorem.
\end{proof}

\end{document}